\documentclass[10pt,conference,hidelinks]{IEEEtran}\IEEEoverridecommandlockouts
\usepackage{epsfig,graphicx,subfigure,psfrag,amsmath,cases}
\usepackage{latexsym,amssymb,amsmath,epsfig,subfigure,algorithm,mathtools}
\usepackage{algorithmic}
\usepackage{color}
\usepackage{url}
\usepackage{scrtime}
\usepackage{cite}
\usepackage{subfigure}
\usepackage{bbding}
\usepackage{multicol}

\author{{Zhiqiang Wei, Lei Yang, Derrick Wing Kwan Ng, and Jinhong Yuan}
\thanks{Zhiqiang Wei, Lei Yang, Derrick Wing Kwan Ng, and Jinhong Yuan are with the School of Electrical Engineering and Telecommunications, the University of New South Wales, Australia (email: zhiqiang.wei@student.unsw.edu.au; lei.yang3@unsw.edu.au; w.k.ng@unsw.edu.au; j.yuan@unsw.edu.au). This work was supported in part by the Australia Research Council (ARC) Discovery Project DP160104566 and Linkage Project LP160100708.}}

\title{On the Performance Gain of NOMA over OMA in Uplink Single-cell Systems}

\newtheorem{Thm}{Theorem}
\newtheorem{proof}{proof}

\newtheorem{T-Prob}{Transformed Problem}

\newtheorem{Remark}{Remark}

\newcommand{\abs}[1]{\lvert#1\rvert}

\begin{document}
\maketitle
\begin{abstract}
In this paper, we investigate the performance gain of non-orthogonal multiple access (NOMA) over orthogonal multiple access (OMA) in uplink single-cell systems.
In both single-antenna and multi-antenna scenarios, the performance gain of NOMA over OMA in terms of asymptotic ergodic sum-rate is analyzed for a sufficiently large number of users.
In particular, in single-antenna systems, we identify two types of near-far gains brought by NOMA:
1) the large-scale near-far gain via exploiting the large-scale fading increases with the cell size; 2) the small-scale near-far gain via exploiting the small-scale fading is a constant given by $\gamma = 0.57721$ nat/s/Hz in Rayleigh fading channels.
Furthermore, we have analyzed that the performance gain achieved by single-antenna NOMA can be amplified via increasing the number of antennas equipped at the base station due to the extra spatial degrees of freedom.
The numerical results confirm the accuracy of the derived analyses and unveil the performance gains of NOMA over OMA in different scenarios.
\end{abstract}

\section{Introduction}
Recently, non-orthogonal multiple access (NOMA) has drawn significant attention in both industry and academia as a promising multiple access technique for the fifth-generation (5G) wireless newtworks\cite{Dai2015,Ding2015b,WeiSurvey2016,wong2017key}.
The principal of NOMA is to exploit the power domain for multiuser multiplexing together with superposition coding and to alleviate inter-user interference (IUI) via successive interference cancellation (SIC) techniques\cite{WeiSurvey2016}.
Towards the evolution of 5G, the industrial community has proposed various forms of NOMA as potential multiple access technologies\cite{ChenNOMAScheme}.

In contrast to conventional orthogonal multiple access (OMA) schemes \cite{DerrickEEOFDMA}, NOMA can serve multiple users via the same degree of freedom (DOF) and achieve a higher spectral efficiency\cite{Sun2016Fullduplex,WeiTCOM2017,Qiu2018Lattice}.
In particular, in single-antenna systems, it has been proved that single-input single-output NOMA (SISO-NOMA) can provide a higher ergodic sum-rate than that of SISO-OMA with a fixed resource allocation scheme\cite{Ding2014}.
Furthermore, the superior performance of SISO-NOMA over SISO-OMA with an optimal resource allocation design in single-antenna systems was shown in \cite{Chen2017}.
In \cite{WangPowerEfficiency,Xu2017}, the authors firstly revealed that the sources of performance gain of NOMA over OMA are two-fold: near-far diversity gain and angle diversity gain when there are multiple antennas.
However, analytical results for quantifying the ergodic sum-rate gain (ESG) of NOMA over OMA has not been reported yet.
More importantly, the potential gain of NOMA over OMA has not been well understood to provide some system design insights.

To achieve a higher spectral efficiency, the concept of NOMA has been extended to multi-antenna systems, namely multiple-input multiple-output NOMA (MIMO-NOMA), based on the signal alignment technique \cite{DingSignalAlignment} and the concept of quasi-degradation \cite{ChenQuasiDegradation}.
In particular, it has been shown that there is a performance gain of MIMO-NOMA over multiple-input multiple-output OMA (MIMO-OMA)\cite{DingSignalAlignment,ChenQuasiDegradation}.
However, when we extend NOMA from single-antenna systems to multi-antenna systems, it is still unclear, if the same ESG can be achieved.
Moreover, how the ESG of NOMA over OMA changes from single-antenna to multi-antenna systems is still an open and interesting problem and deserves our efforts to explore.
The answers to these questions above are the key to unlock the potential in applying NOMA to the future 5G communications systems.

This paper aims to deepen the understanding on the ESG of NOMA over OMA in uplink single-cell systems.
In particular, we quantify the ESG of NOMA over OMA in both single-antenna and multi-antenna systems.
Furthermore, performance analyses are provided to study the improvement of ESG when extending NOMA from single-antenna systems to multi-antenna systems.
Simulation results confirm the accuracy of our derived performance analyses and demonstrate some interesting insights which are summarized in the following:
\begin{itemize}
  \item In both single-antenna and multi-antenna scenarios, the ESGs of NOMA over OMA saturate for large numbers of users and sufficiently high signal-to-noise ratio (SNR).
  \item In single-antenna scenario, we identify two types of \emph{near-far gains}\cite{Xu2017} brought by NOMA and reveal their different behaviors.
      In particular, the \emph{large-scale near-far gain} via exploiting the large-scale fading increases with the cell size, while the \emph{small-scale near-far gain} arising from the small-scale fading is a constant of $\gamma = 0.57721$ nat/s/Hz in Rayleigh fading channels.
  \item As applying NOMA in multi-antenna systems, the ESG of SISO-NOMA over SISO-OMA can be amplified $M$-fold when the base station is equipped with $M$ antennas owing to the \emph{spatial degrees of freedom}\cite{Xu2017}.
\end{itemize}

Notations used in this paper are as follows. Boldface capital and lower case letters are reserved for matrices and vectors, respectively. ${\left( \cdot \right)^{\mathrm{T}}}$ denotes the transpose of a vector or matrix and ${\left( \cdot \right)^{\mathrm{H}}}$ denotes the Hermitian transpose of a vector or matrix.
$\mathbb{C}^{M\times N}$ denotes the set of all $M\times N$ matrices with complex entries;
$\abs{\cdot}$ denotes the absolute value of a complex scalar or the determinant of a matrix, and $\|{\cdot}\|$ denotes Euclidean norm of a complex vector.
The circularly symmetric complex Gaussian distribution with mean $\mu$ and variance $\sigma^2$ is denoted by ${\cal CN}(\mu,\sigma^2)$.

\section{System Model}
\begin{figure}[t]
\centering
\includegraphics[width=2in]{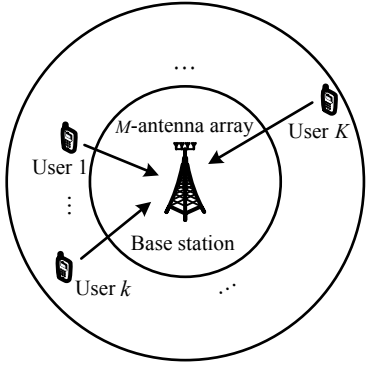}\vspace{-3mm}
\caption{The system model of uplink communication with one BS and $K$ users.}\vspace{-5mm}
\label{NOMA_Uplink_Model}
\end{figure}
\subsection{System Model}
We consider an uplink\footnote{We restrict ourselves to the uplink NOMA communications, as the reception of NOMA is more affordable for the base station.} single-cell NOMA system with a base station (BS) and $K$ users, as shown in Fig. \ref{NOMA_Uplink_Model}.
The cell is modeled as two concentric ring-shaped discs.
The BS is located at the center of the ring-shaped discs with the inner radius of $D_0$ and outer radius of $D$, wherein all the $K$ users are scattered uniformly.
For the NOMA scheme, all the $K$ users are multiplexed on the same frequency band and time slot, while for the OMA scheme, $K$ users utilize the frequency or time resources orthogonally.
Without loss of generality, we consider a frequency division multiple access (FDMA) as a typical OMA scheme.

As shown in Fig. \ref{NOMA_Uplink_Model}, we consider two kinds of typical communication systems:
\begin{itemize}
  \item SISO-NOMA and SISO-OMA: the BS is equipped with a single-antenna ($M=1$) and all the $K$ users are single-antenna devices.
  \item MIMO-NOMA and MIMO-OMA: the BS is equipped with a multi-antenna array ($M>1$) and all the $K$ users are single-antenna devices with $K > M$.
\end{itemize}
Note that, in the second case, we still denote them as MIMO-NOMA and MIMO-OMA even each user is equipped with a single-antenna, because there are multiple users uploading their independent data streams to the BS.

The received signal at the BS is given by
\vspace{-1.5mm}
\begin{equation}\label{MIMONOMASystemModel}
{\bf{y}} = \sum\nolimits_{k = 1}^K {{{\bf{h}}_k}} \sqrt {{p_k}} {x_k} + {\bf{v}},
\vspace{-1.5mm}
\end{equation}
where ${\bf{y}}\in \mathbb{C}^{ M \times 1}$, $p_k$ denotes the power allocation for user $k$, $x_k$ denotes the normalized modulated symbol for user $k$ with ${\mathrm{E}}\left\{ \left|{x_k}\right|^2 \right\} = 1$, and ${\bf{v}}\sim \mathcal{CN}\left(\mathbf{0},N_0 {{\bf{I}}_M}\right)$ denotes the additive white Gaussian noise (AWGN) at the BS with zero mean and covariance matrix of $N_0 {{\bf{I}}_M}$.
For the system power budget, we consider a sum-power constraint for all the users, i.e.,
\vspace{-1.5mm}
\begin{equation}\label{SumPowerConstraint}
\sum\nolimits_{k = 1}^K {{p_{k}}}  \le {P_{\text{max}}},
\vspace{-1.5mm}
\end{equation}
where ${P_{\text{max}}}$ is the maximum transmit power for all the users.
Note that, the sum-power constraint is a commonly adopted assumption in the literature to simplify the performance analysis for uplink communications, e.g. \cite{Vishwanath2003,WangMUG,Xu2017}.
In fact, the sum-power constraint is a reasonable assumption for practical cellular communication systems, where a total transmit power limitation is usually imposed to prevent a large inter-cell interference.

The channel vector between user $k$ and the BS is modeled as
\vspace{-1mm}
\begin{equation}\label{ChannelModel}
{{{\bf{h}}_k}} = \frac{{\bf{g}}_k}{\sqrt{1+d_k^{\alpha}}},
\vspace{-1mm}
\end{equation}
where ${\bf{g}}_k \in \mathbb{C}^{ M \times 1}$ denotes the Rayleigh fading coefficients, i.e., ${\bf{g}}_k \sim \mathcal{CN}\left(\mathbf{0},{{\bf{I}}_M}\right)$, $d_k$ denotes distance between user $k$ and the BS with the unit of meter, and $\alpha$ denotes the path loss exponent\footnote{In this paper, we ignore the log-normal shadowing to simplify our performance analysis.
Note that, the log-normal shadowing only introduces an additional scalar multiplication to the channel model in \eqref{ChannelModel}.
Although the log-normal shadowing may change the resulting channel distribution of ${{{\bf{h}}_k}}$, the distance-based channel model is sufficient to characterize the near-far gain \cite{Xu2017} discussed in this paper.}.
We note ${\bf{H}} = \left[ {{{\bf{h}}_1}, \ldots ,{{\bf{h}}_K}} \right] \in \mathbb{C}^{ M \times K}$ as the channel matrix between all the $K$ users and the BS.
When $M=1$, ${{{{h}}_k}} = \frac{{{g}}_k}{\sqrt{1+d_k^{\alpha}}}$ denotes the corresponding channel coefficient of user $k$ in single-antenna systems.
We assume that the channel coefficients are independent and identically distributed (i.i.d.) over all the users and antennas.
Since this paper aims to provide some theoretical insights about the performance gain of NOMA over OMA, we assume that perfect channel state information (CSI) is known at the BS.
Without loss of generality, we assume $\left\| {{{\bf{h}}_1}} \right\| \ge, \ldots, \ge\left\| {{{\bf{h}}_K}} \right\|$, i.e., users are indexed based on their channel gains.
\subsection{Signal Detection and Resource Allocation Strategy}
To facilitate the performance analysis, we focus on the following well-known and effective signal detection and resource allocation strategies.
\subsubsection{Signal detection}
For SISO-NOMA, we adopt the commonly used SIC \cite{wei2017performance} decoding at the BS, since it is capacity achieving for single-antenna systems\cite{Tse2005}.
On the other hand, for SISO-OMA, as all the users are separated by orthogonal frequency subbands and thus the simple single-user detection (SUD) can be used to achieve the optimal performance.
For MIMO-NOMA, the minimum mean square error and successive interference cancellation (MMSE-SIC) decoding is an appealing reception technique since it is capacity achieving \cite{Tse2005} with an acceptable computational complexity for a finite number of antennas $M$ at the BS.
On the other hand, FDMA zero forcing (FDMA-ZF) is adopted for MIMO-OMA, which is known as capacity achieving in the high SNR regime\cite{Tse2005}.
In particular, owing to the extra spatial DOF induced by multiple antennas at the BS, all the users can be divided into $G = K/M$ groups\footnote{Without loss of generality, we consider the case with $G$ as an integer in this paper.} with each group containing $M$ users.
Then, ZF is utilized to handle the inter-user interference within each group and FDMA is utilized to separate all the $G$ groups on orthogonal frequency subbands.

\subsubsection{Resource allocation strategy}
We consider an equal resource allocation (ERA) strategy for both NOMA and OMA scheme.
In particular, equal power allocation is adopted for NOMA schemes, i.e., ${{p_{k}}} = \frac{{P_{\text{max}}}}{K}$.
On the other hand, equal power and frequency allocation is adopted for OMA schemes, i.e., ${{p_{k}}} = \frac{{P_{\text{max}}}}{K}$ and $f_k = 1/K$ or $f_g = 1/G$, where $f_k$ and $f_g$ denotes the normalized frequency allocation for user $k$ in SISO-OMA and for group $g$ in MIMO-OMA, respectively.
Note that the ERA is a typical but effective strategy utilized in some application scenarios requiring a limited system overhead, e.g. machine-type communications (MTC).

We note that the user grouping design is involved in the considered MIMO-OMA system with FDMA-ZF.
Fortunately, for any user grouping strategy, the ergodic sum-rate of MIMO-OMA remains the same, since all the users have i.i.d. channel distributions and can access the identical resource.
Therefore, a random user grouping strategy is adopted for the MIMO-OMA system with FDMA-ZF.
In particular, we randomly select $M$ users for each group on each frequency subband and denote the corresponding composite channel matrix of $g$-th group as ${\bf{H}}_g = \left[ {{{\bf{h}}_{(g-1)M+1}}, \ldots ,{{\bf{h}}_{gM}}} \right] \in \mathbb{C}^{ M \times M}$.
On the other hand, we consider a fixed SIC and MMSE-SIC decoding order at the BS as $1,\ldots,K$ for SISO-NOMA and MIMO-NOMA, respectively\footnote{Note that the SIC and MMSE-SIC decoding order at the BS do not affect the system sum-rate in uplink SISO-NOMA systems\cite{WeiLetter2018} and uplink MIMO-NOMA systems\cite{Tse2005}, respectively.}.
In addition, to provide some theoretical insights about the performance gain of NOMA over OMA, we assume that there is no error propagation during SIC and MMSE-SIC decoding at the BS.
\section{ESG of SISO-NOMA over SISO-OMA}
In this section, we first derive the ergodic sum-rate of SISO-NOMA and SISO-OMA.
Then, the ESG of SISO-NOMA over SISO-OMA is discussed asymptotically.

\subsection{Ergodic Sum-rate}
The instantaneous achievable data rate of user $k$ of the considered SISO-NOMA system is given by:
\vspace{-1mm}
\begin{equation}\label{SISONOMAIndividualAchievableRate}
R_{k}^{\text{SISO-NOMA}} = {\ln}\left( 1 + \frac{{{p_k}{{\left| {{h_k}} \right|}^2}}}{{\sum\nolimits_{i = k + 1}^K {{p_i}{{\left| {{h_i}} \right|}^2}}  + {N_0}}} \right),
\vspace{-1mm}
\end{equation}
while the instantaneous achievable data rate of user $k$ in the considered SISO-OMA system is given by:
\vspace{-1mm}
\begin{equation}\label{SISOOMAIndividualAchievableRate}
R_{k}^{{\rm{SISO-OMA}}} = f_k {\ln}\left(1+ {\frac{{{p_{k }}{{\left| {{{{h}}_k}} \right|}^2}}}{{f_k N_0}}} \right).
\vspace{-1mm}
\end{equation}
Under the ERA strategy defined above, we have the instantaneous sum-rate of SISO-NOMA and SISO-OMA given by
\begin{align}
R_{\rm{sum}}^{{\rm{SISO-NOMA}}} &= \sum\nolimits_{k = 1}^K R_{k}^{\text{SISO-NOMA}}\notag\\
&=  {\ln}\left( 1 + \frac{{P_{\text{max}}}}{K {N_0}} \sum\nolimits_{k = 1}^K {{\left| {{{h}_k}} \right|}^2} \right) \;\text{and} \label{InstantSumRateSISONOMA}\\
R_{\rm{sum}}^{{\rm{SISO-OMA}}} &= \sum\nolimits_{k = 1}^K R_{k}^{\text{SISO-OMA}}\notag\\
&= \frac{1}{K}\sum\nolimits_{k = 1}^K {\ln}\left( 1 + \frac{P_{\text{max}}}{{N_0}} {{\left| {{{h}_k}} \right|}^2} \right), \label{InstantSumRateSISOOMA}
\end{align}
respectively.
It can be observed that in the low SNR regime with ${P_{\text{max}}} \to 0$, the performance gain of SISO-NOMA over SISO-OMA vanishes.

Since all the users are scattered uniformly in two concentric rings with the inner radius of $D_0$ and outer radius of $D$, the cumulative distribution function (CDF) of the channel gain ${{\left| {{h}} \right|}^2}$ is given by\footnote{Since all the users have i.i.d. channel distributions, the subscript $k$ might be dropped without causing notational confusion in this paper.}
\vspace{-1mm}
\begin{equation}\label{ChannelDistributionCDF}
{F_{{{\left| {{h}} \right|}^2}}}\left( x \right) = \int_{D_0}^D {\left( {1 - {e^{ - \left( {1 + {z^\alpha }} \right)x}}} \right)} {f_{{d}}}\left( z \right)dz,
\vspace{-1mm}
\end{equation}
where ${f_{{d}}}\left( z \right) = \frac{2z}{D^2 - D_0^2}$, $D_0 \le z \le D$, denotes the probability density function (PDF) for the random distance $d$.
With the Gaussian-Chebyshev quadrature approximation\cite{abramowitz1964handbook}, the CDF and PDF of ${{\left| {{h}} \right|}^2}$ are given by
\begin{align}
{F_{{{\left| {{h}} \right|}^2}}}\left( x \right) &\approx 1 - \frac{1}{D+D_0}\sum\nolimits_{n = 1}^N {{\beta _n}{e^{ - {c_n}x}}} \; \text{and} \label{SISOChannelDistributionCDF}\\
{f_{{{\left| {{h}} \right|}^2}}}\left( x \right) &\approx \frac{1}{D+D_0}\sum\nolimits_{n = 1}^N {{\beta _n}{c_n}{e^{ - {c_n}x}}}, x \ge 0, \label{SISOChannelDistributionPDF}
\end{align}
respectively, where
\begin{align}\label{BetaCn}
{\beta_n} &= \frac{\pi }{N}\left| {\sin \frac{{2n \hspace{-1mm}-\hspace{-1mm} 1}}{{2N}}\pi } \right|\left( {\frac{D\hspace{-1mm}-\hspace{-1mm}D_0}{2}\cos \frac{{2n \hspace{-1mm}-\hspace{-1mm} 1}}{{2N}}\pi  + \frac{D\hspace{-1mm}+\hspace{-1mm}D_0}{2}} \right) \;\text{and}\notag\\
{c_n} &= 1 + {\left( {\frac{D\hspace{-1mm}-\hspace{-1mm}D_0}{2}\cos \frac{{2n \hspace{-1mm}-\hspace{-1mm} 1}}{{2N}}\pi  + \frac{D\hspace{-1mm}+\hspace{-1mm}D_0}{2}} \right)^\alpha }
\end{align}
are approximation parameters and $N$ denotes the number of integral approximation terms.

Based on \eqref{InstantSumRateSISONOMA}, we have the asymptotic ergodic sum-rate of the considered SISO-NOMA system with $K\rightarrow \infty$ as follows:
\begin{align}\label{ErgodicSumRateSISONOMA}
\hspace{-2mm}\mathop {\lim }\limits_{K \to \infty }  \overline{R_{\rm{sum}}^{{\rm{SISO-NOMA}}}}
& \mathop = \limits^{(a)} {\ln}\left( 1 + \frac{P_{\text{max}}}{N_0} \overline{{{\left| {{h}} \right|}^2}} \right) \notag\\
& \approx {\ln}\left(\hspace{-1mm} {1 \hspace{-1mm}+ \hspace{-1mm} \frac{P_{\text{max}}}{{{\left(D\hspace{-1mm}+\hspace{-1mm}D_0\right)}N_0}}\sum\nolimits_{n = 1}^N \hspace{-1mm} {\frac{{{\beta _n}}}{{{c_n}}}} } \hspace{-1mm}\right),
\end{align}
where $\overline{{{\left| {{h}} \right|}^2}}$ denotes the average channel power gain and it is given by
\vspace{-1mm}
\begin{equation}\label{Mean_ChannelPowerGainSISO}
\overline{{{\left| {{h}} \right|}^2}} \approx \int_0^\infty  {x{f_{{{\left| {{h}} \right|}^2}}}\left( x \right)} dx
= \frac{1}{D+D_0}\sum\nolimits_{n = 1}^N {\frac{{{\beta _n}}}{{{c_n}}}}.
\vspace{-1mm}
\end{equation}
Note that the equality $(a)$ in \eqref{ErgodicSumRateSISONOMA} only holds asymptotically for $K \to \infty$ due to $\mathop {\lim } \limits_{K \to \infty } {\frac{1}{K}\sum\nolimits_{k = 1}^K {{{\left| {{h_k}} \right|}^2}} } = \overline{{{\left| {{h}} \right|}^2}}$.
Otherwise, for a finite number of users $K$, the asymptotic ergodic sum-rate in \eqref{ErgodicSumRateSISONOMA} serves as an upper bound for the actual ergodic sum-rate, i.e., $\mathop {\lim }\limits_{K \to \infty }  \overline{R_{\rm{sum}}^{{\rm{NOMA}}}} \ge \overline{R_{\rm{sum}}^{{\rm{NOMA}}}}$, owing to the concavity of the logarithm function and the Jensen's inequality.
In the simulations, we show that the asymptotic analysis in \eqref{ErgodicSumRateSISONOMA} is accurate even for a finite value of $K$ and becomes tight with increasing $K$.

For the considered SISO-OMA system, based on \eqref{InstantSumRateSISOOMA}, we obtain its ergodic sum-rate as follows:
\begin{align}\label{ErgodicSumRateSISOOMA}
\hspace{-2mm}\overline{R_{\rm{sum}}^{{\rm{SISO-OMA}}}}
& = \int_0^\infty  {{{\ln }}\left( {1 + \frac{P_{\text{max}}}{N_0}x} \right){f_{{{\left| {{h}} \right|}^2}}}\left( x \right)} dx \notag\\
& = \frac{1}{\left(D\hspace{-1mm}+\hspace{-1mm}D_0\right)}\sum\nolimits_{n = 1}^N {{\beta _n}{e^{\frac{{{c_n N_0}}}{{P_{\text{max}}}}}} {E_1}\left( {\frac{{{c_n N_0}}}{{P_{\text{max}}}}} \right)},
\end{align}
where ${E_l}\left( x \right) = \int_1^\infty  {\frac{{{e^{ - xt}}}}{t^l}} dt$ denotes the $l$-order exponential integral\cite{abramowitz1964handbook}.
Note that, different from $\overline{R_{\rm{sum}}^{{\rm{SISO-NOMA}}}}$ in \eqref{ErgodicSumRateSISONOMA}, $\overline{R_{\rm{sum}}^{{\rm{SISO-OMA}}}}$ in \eqref{ErgodicSumRateSISOOMA} is applicable for an arbitrary number of users $K$.

\subsection{ESG in Single-antenna Systems}
Based on \eqref{ErgodicSumRateSISONOMA} and \eqref{ErgodicSumRateSISOOMA}, the asymptotic ESG of SISO-NOMA over SISO-OMA with ${K \rightarrow \infty}$ can be given as follows:
\begin{align}\label{EPGSISOERA}
\hspace{-2mm}\mathop {\lim }\limits_{K \rightarrow \infty}  \overline{G^{\rm{SISO}}} &= \mathop {\lim }\limits_{K \to \infty }  \overline{R_{\rm{sum}}^{{\rm{SISO-NOMA}}}} - \overline{R_{\rm{sum}}^{{\rm{SISO-OMA}}}} \notag\\
&\approx {\ln}\left( {1 + \frac{{P_{\text{max}}}}{{{\left(D+D_0\right)N_0}}}\sum\nolimits_{n = 1}^N {\frac{{{\beta _n}}}{{{c_n}}}} } \right)\notag\\
&- \frac{1}{\left(D\hspace{-1mm}+\hspace{-1mm}D_0\right)}\sum\nolimits_{n = 1}^N {{\beta _n}{e^{\frac{{{c_n}N_0}}{{P_{\text{max}}}}}} {E_1}\left( {\frac{{{c_n}N_0}}{{P_{\text{max}}}}} \right)}.
\end{align}
To obtain more insights, in the high SNR regime with ${P_{\text{max}}} \rightarrow \infty$, we approximate the asymptotic ESG in \eqref{EPGSISOERA} as
\vspace{-1mm}
\begin{equation}\label{EPGSISOERA2}
\mathop {\lim }\limits_{K \rightarrow \infty, {P_{\text{max}}} \rightarrow \infty}  \overline{G^{\rm{SISO}}} \approx \vartheta \left( {D,{D_0}} \right) + \gamma,
\vspace{-1mm}
\end{equation}
where $\vartheta \left( {D,{D_0}} \right)$ denotes the large-scale near-far gain and it is given by
\vspace{-1mm}
\begin{equation}\label{NearFarDiversity}
\vartheta \left( {D,{D_0}} \right)
 = \ln \left( {\frac{{\sum\limits_{n = 1}^N {\left( {\frac{1}{{{c_n}}}} \right)\frac{{{\beta _n}}}{{\left( {D + {D_0}} \right)}}} }}{{\mathop \Pi \limits_{n = 1}^N {{\left( {\frac{1}{{{c_n}}}} \right)}^{\frac{{{\beta _n}}}{{\left( {D + {D_0}} \right)}}}}}}} \right).
\vspace{-1mm}
\end{equation}
The Euler-Mascheroni constant, $\gamma \approx 0.57721$ \cite{abramowitz1964handbook}, denotes the small-scale near-far gain in Rayleigh fading channels.
The approximation in \eqref{EPGSISOERA2} is obtained by applying $\mathop {\lim }\limits_{x \to 0} {E_1}\left( x \right)\hspace{-1mm} \approx \hspace{-1mm} - \hspace{-1mm}\ln \left( x \right)\hspace{-1mm} - \hspace{-1mm}\gamma $ \cite{abramowitz1964handbook}.

\begin{figure}[t]
\vspace{-4mm}
\centering
\includegraphics[width=2.8in]{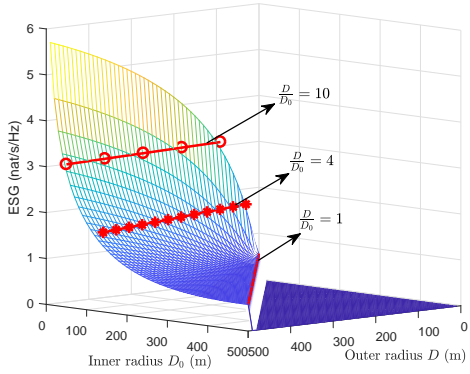}\vspace{-3mm}
\caption{The asymptotic ESG in \eqref{EPGSISOERA2} under ERA versus $D$ and $D_0$ with $K \rightarrow \infty$ and ${P_{\text{max}}} \rightarrow \infty$.}\vspace{-5mm}
\label{ESG_ERA}
\end{figure}

There are three important observations from \eqref{EPGSISOERA2}. Firstly, it can be observed that the asymptotic ESG in \eqref{EPGSISOERA2} is not a function of ${P_{\text{max}}}$ but a function of $D$ and $D_0$.
It implies that the asymptotic ESG saturates for a sufficiently large ${P_{\text{max}}}$.
In fact, the sum-rates of both SISO-NOMA systems in \eqref{InstantSumRateSISONOMA} and SISO-OMA systems in \eqref{InstantSumRateSISOOMA} increase very slowly with ${P_{\text{max}}} \rightarrow \infty$, which results in a saturated ESG in the high SNR regime.

Secondly, based on the weighted arithmetic and geometric means (AM-GM) inequality\cite{kedlaya1994proof}, we can show that $\vartheta \left( {D,{D_0}} \right) \ge 0$.
In particular, for the extreme case that all the users are randomly deployed on a circle, i.e., $D = D_0$, we have $c_1 = ,\ldots, = c_N$ according to \eqref{BetaCn} and $\vartheta \left( {D,{D_0}} \right) = 0$ and thus the performance gain achieved by NOMA is bounded by below with $\gamma = 0.57721$ nat/s/Hz.
Otherwise, the asymptotic ESG is larger than $\gamma$ for a general cell deployment $D > D_0$.
As a result, in the asymptotic case with $K \rightarrow \infty$ and ${P_{\text{max}}} \rightarrow \infty$, SISO-NOMA can provide \emph{at least $0.57721$} nat/s/Hz spectral efficiency gain over SISO-OMA for an arbitrary user deployment.
In fact, the minimum asymptotic ESG gain $\gamma$ arises from the small-scale Rayleigh fading since we force all the users have the same large-scale fading when $D = D_0$.

Thirdly, the dominating component of the asymptotic ESG \eqref{EPGSISOERA2} originates from function $\vartheta \left( {D,{D_0}} \right)$ given in \eqref{NearFarDiversity}, which actually characterizes the near-far gain due to large-scale fading.
To visualize the large-scale near-far gain, we illustrate the asymptotic ESG in \eqref{EPGSISOERA2} versus $D$ and $D_0$ in Fig. \ref{ESG_ERA}.
We can observe that, when $D = D_0$,  the minimum asymptotic ESG $\gamma \approx 0.57721$ nat/s/Hz due to the small-scale near-far gain can be obtained.
Besides, when $D > D_0$, the asymptotic ESG is much larger than $\gamma$ due to the large-scale near-far gain $\vartheta \left( {D,{D_0}} \right)$.
It can be observed that $\vartheta \left( {D,{D_0}} \right)$ in \eqref{NearFarDiversity} is actually a function of $\frac{D}{D_0}$.
It implies that, the larger $\frac{D}{D_0}$, the larger $\vartheta \left( {D,{D_0}} \right)$ and the larger ESG.
In fact, for a larger $\frac{D}{D_0}$, the heterogeneity in the large-scale fading among users becomes higher and the SISO-NOMA can exploit the near-far diversity more efficiently to improve the sum-rate performance.

\begin{Remark}
Note that, in the literature \cite{Dingtobepublished,WeiTCOM2017}, it has been shown that two users with a large distance difference or channel gain difference are preferred to be paired.
It is consistent with our conclusion in this paper, where a larger $\frac{D}{D_0}$ implies a higher ESG of NOMA over OMA.
However, it is worth to note that this is the first work which quantifies the ESG of NOMA over OMA and identifies two kinds of near-far gains in ESG.
\end{Remark}
\section{ESG of MIMO-NOMA over MIMO-OMA}
In this section, we first derive the ergodic sum-rate of MIMO-NOMA and MIMO-OMA.
Then, the ESG of MIMO-NOMA over MIMO-OMA is discussed asymptotically.

\subsection{Ergodic Sum-rate}
The instantaneous achievable data rate of user $k$ of the considered MIMO-NOMA system with the MMSE-SIC detection is given by\cite{Tse2005}:
\begin{align}\label{MIMONOMAIndividualAchievableRate}
R_{k}^{{\rm{MIMO-NOMA}}} = \ln \left| {{{\bf{I}}_M} + \frac{1}{{{N_0}}}\sum\limits_{i = k}^K {{p_i}{{\bf{h}}_i}{\bf{h}}_i^{\rm{H}}} } \right| \notag\\
- \ln \left| {{{\bf{I}}_M} + \frac{1}{{{N_0}}}\sum\limits_{i = k + 1}^K {{p_i}{{\bf{h}}_i}{\bf{h}}_i^{\rm{H}}} } \right|.
\end{align}

For the considered MIMO-OMA system, the instantaneous achievable data rate of user $k$ is given by \cite{Tse2005}
\vspace{-1mm}
\begin{equation}\label{MIMOOMAZFIndividualAchievableRate}
R_{k}^{{\rm{MIMO-OMA-ZF}}} = f_g {\ln}\left(1+ {\frac{{{p_{k }}{{\left| {{{\mathbf{w}_{g,k}^{\rm{H}}\mathbf{h}}_k}} \right|}^2}}}{{f_g N_0}}} \right),
\vspace{-1mm}
\end{equation}
where vector $\mathbf{w}_{g,k} \in \mathbb{C}^{ M \times 1}$ denotes the normalized ZF detection vector for user $k$ with ${\left\| {{{\mathbf{w}_{g,k}}}} \right\|}^2 = 1$, which is obtained based on the pseudoinverse of the channel matrix ${\bf{H}}_g$ in the $g$-th user group\cite{Tse2005}.

With the ERA strategy, the instantaneous sum-rate of MIMO-NOMA and MIMO-OMA are obtained by
\begin{align}
\hspace{-2mm}R_{\rm{sum}}^{{\rm{MIMO-NOMA}}} & = \sum\nolimits_{k = 1}^K R_{k}^{{\rm{MIMO-NOMA}}} \notag\\
&= {\ln}\left| {{{\bf{I}}_M} + \frac{{P_{\text{max}}}}{{{KN_0}}}\sum\nolimits_{k = 1}^K {{{\bf{h}}_k}{\bf{h}}_k^{\rm{H}}} } \right| \;\text{and}\label{InstantSumRateMIMONOMA} \\
\hspace{-2mm}R_{\rm{sum}}^{{\rm{MIMO-OMA}}}
& = \sum\nolimits_{k = 1}^K R_{k}^{{\rm{MIMO-OMA}}} \notag\\
&= \frac{M}{K}\sum\nolimits_{k = 1}^K \hspace{-1mm}{\ln}\left(1\hspace{-1mm}+ \hspace{-1mm} \frac{{P_{\text{max}}}}{MN_0}{\left| {{{\mathbf{w}_{g,k}^{\rm{H}}\mathbf{h}}_k}} \right|}^2 \right), \label{InstantSumRateMIMOOMAZF}
\end{align}
respectively.

In fact, MMSE-SIC is capacity achieving \cite{Tse2005} and \eqref{InstantSumRateMIMONOMA} is the channel capacity with the deterministic channel matrix ${\bf{H}}$\cite{GoldsmithMIMOCapacity2003}.
It is difficult to obtain a closed-form expression for the channel capacity above due to the determinant of summation of matrices in \eqref{InstantSumRateMIMONOMA}.
To provide more insights, in the following theorem, we consider an asymptotically tight upper bound for the achievable sum-rate in \eqref{InstantSumRateMIMONOMA} with $K \to \infty$.

\begin{Thm}\label{Theorem1}
For the considered MIMO-NOMA system in \eqref{MIMONOMASystemModel} with the MMSE-SIC detection, the achievable sum-rate in \eqref{InstantSumRateMIMONOMA} is upper bounded by above as
\vspace{-1mm}
\begin{equation}\label{InstantSumRateMIMONOMA_UpperBound}
\hspace{-2mm}R_{\rm{sum}}^{{\rm{MIMO-NOMA}}} \le M{\ln}\left( {1 + \frac{{P_{\text{max}}}}{{{KMN_0}}}\sum\nolimits_{k = 1}^K {{{\left\| {{{\bf{h}}_k}} \right\|}^2}} } \right).
\vspace{-1mm}
\end{equation}
Also, the upper bound is asymptotically tight when $K \to \infty$, i.e.,
\vspace{-1mm}
\begin{equation}\label{AsympInstantSumRateMIMONOMA}
\hspace{-2mm}\mathop {\lim }\limits_{K \rightarrow \infty} \hspace{-1mm} R_{\rm{sum}}^{{\rm{MIMO-NOMA}}}
\hspace{-0.5mm}=\hspace{-0.5mm} \mathop {\lim }\limits_{K \rightarrow \infty} M{\ln}\hspace{-0.5mm}\left( \hspace{-0.5mm}{1 \hspace{-0.5mm}+\hspace{-0.5mm} \frac{{P_{\text{max}}}}{{{KMN_0}}}\sum\limits_{k = 1}^K {{{\left\| {{{\bf{h}}_k}} \right\|}^2}} } \hspace{-0.5mm}\right)\hspace{-0.5mm}.
\vspace{-1mm}
\end{equation}
\end{Thm}
\begin{proof}
The main idea of proof is based on the bounding techniques in \cite{WangMUG}.
Due to the page limitation, we omit the detailed proof here and leave it for the journal version.
\end{proof}

Given the distance from a user to the BS as $d$, its channel gain ${{\left\| {{\mathbf{h}}} \right\|}^2}$ follows the Gamma distribution\cite{yang2017noma}, whose PDF and CDF are given by
\begin{align}
{f_{{{\left\| {{\mathbf{h}}} \right\|}^2 | d}}}\left( x \right) &= \text{Gamma} \left(M, 1+ d^\alpha\right)\; \text{and}\notag\\
{F_{{{\left\| {{\mathbf{h}}} \right\|}^2 | d}}}\left( x \right) &= \frac{{\gamma_{L} \left( {M,\left( {1 + d^\alpha } \right)x} \right)}}{\Gamma \left( M \right)},
\end{align}
respectively, where $\text{Gamma} \left(M, \xi\right) = \frac{{{{ \xi }^M}{x^{M - 1}}{e^{ -  \xi x}}}}{\Gamma \left( M \right)}$ denotes the PDF of a random variable with a Gamma distribution, ${\Gamma \left( M \right)}$ denotes the Gamma function, and ${\gamma_{L} \left( {M,\left( {1 + d^\alpha } \right)x} \right)}$ denotes the lower incomplete Gamma function.

Then, the CDF of the channel gain ${{\left\| {{\mathbf{h}}} \right\|}^2}$ is given by
\vspace{-1mm}
\begin{equation}
{F_{{{\left\| {{\mathbf{h}}} \right\|}^2}}}\left( x \right) = \int_{D_0}^D \frac{{\gamma_{L} \left( {M,\left( {1 + d^\alpha } \right)x} \right)}}{{\Gamma \left(M\right)}} {f_{{d}}}\left( z \right)dz.
\vspace{-1mm}
\end{equation}
Again, with the Gaussian-Chebyshev quadrature approximation\cite{abramowitz1964handbook}, the CDF and PDF of ${{\left\| {{\mathbf{h}}} \right\|}^2}$ are given by
\begin{align}
\hspace{-2mm}{F_{{{\left\| {{\mathbf{h}}} \right\|}^2}}}\left( x \right) &\approx 1 - \frac{1}{D+D_0}\sum\nolimits_{n = 1}^N \frac{{{\beta _n}\gamma_{L} \left( {M,{c_n}x} \right)}}{{\Gamma \left(M\right)}} \; \text{and} \notag\\
\hspace{-2mm}{f_{{{\left\| {{\mathbf{h}}} \right\|}^2}}}\left( x \right) &\approx \frac{1}{D+D_0}\sum\nolimits_{n = 1}^N {{\beta _n} \text{Gamma} \left(M, c_n\right)}, x \ge 0,
\end{align}
respectively, where ${\beta _n}$ and ${c_n}$ are given by \eqref{BetaCn}.

According to \eqref{AsympInstantSumRateMIMONOMA}, the asymptotic ergodic sum-rate of MIMO-NOMA with $K\rightarrow \infty$ is given by
\begin{align} \label{ErgodicSumRateMIMONOMA}
\mathop {\lim }\limits_{K \to \infty } \overline{R_{\rm{sum}}^{{\rm{MIMO-NOMA}}}}
& = M{\ln}\left( 1 + \frac{{P_{\text{max}}}}{MN_0} \overline{{{\left\| {{\mathbf{h}}} \right\|}^2}} \right), \\
& \approx M{\ln}\left(\hspace{-1mm}{1+ \frac{{P_{\text{max}}}}{{{\left(D\hspace{-1mm}+\hspace{-1mm}D_0\right)N_0}}}
\sum\nolimits_{n = 1}^N \hspace{-1mm}{\frac{{{\beta _n}}}{{{c_n}}}} } \hspace{-1mm}\right),\notag
\end{align}
where $\overline{{{\left\| {{\mathbf{h}}} \right\|}^2}}$ denotes the average channel gain and it is given by
\vspace{-1mm}
\begin{equation}\label{Mean_ChannelPowerGainMIMO}
\overline{{{\left\| {{\mathbf{h}}} \right\|}^2}} \approx \int_0^\infty  {x{f_{{{\left\| {{\mathbf{h}}} \right\|}^2}}}\left( x \right)} dx
= \frac{M}{D+D_0}\sum\nolimits_{n = 1}^N {\frac{{{\beta _n}}}{{{c_n}}}}.
\vspace{-1mm}
\end{equation}

\begin{Remark} \label{Remark2}
Comparing \eqref{ErgodicSumRateSISONOMA} and \eqref{ErgodicSumRateMIMONOMA}, we can observe that the performance of a MIMO-NOMA system is asymptotically equivalent to that of a SISO-NOMA system with $M$-fold orthogonal frequency bands and an equivalent channel gain of ${{\left\| {{{\bf{h}}}} \right\|}^2}$, when there is a sufficiently large number of users.
On the other hand, for $K \to \infty$ but a finite $M$, the number of antennas at the BS is much smaller compared to the number of users.
In such a case, the performance of this system approaches the one with a single-antenna BS.
In addition, with a sufficiently large number of users, the received signals at the BS fully span the $M$-dimensional signal space\cite{WangMUG}.
Therefore, MIMO-NOMA with MMSE-SIC can fully utilize the system spatial DOF, $M$, and its performance can be approximated by that of a SISO-NOMA system with $M$-fold frequency bands.
\end{Remark}

For MIMO-OMA, since ${\left\| {{{\mathbf{w}_{g,k}}}} \right\|}^2 = 1$ and ${\bf{g}}_k \sim \mathcal{CN}\left(\mathbf{0},{{\bf{I}}_M}\right)$, we have ${\left| {{{\mathbf{w}_{g,k}^{\rm{H}}\mathbf{g}}_k}} \right|} \sim \mathcal{CN}\left({0},1\right)$ \cite{Tse2005}.
As a result, ${\left| {{{\mathbf{w}_{g,k}^{\rm{H}}\mathbf{h}}_k}} \right|}$ in \eqref{InstantSumRateMIMOOMAZF} has an identical distribution with ${{\left| {{h}} \right|}^2}$, and its CDF and PDF can be given by \eqref{SISOChannelDistributionCDF} and \eqref{SISOChannelDistributionPDF}, respectively.
Therefore, the ergodic sum-rate of the considered MIMO-OMA system can be obtained by
\begin{align}\label{ErgodicSumRateMIMOOMAZF}
\hspace{-2mm}\overline{R_{\rm{sum}}^{{\rm{MIMO-OMA}}}} &= \int_0^\infty  M{{{\ln }}\left( {1 + \frac{P_{\text{max}}}{MN_0}x} \right){f_{{{\left| {{h}} \right|}^2}}}\left( x \right)} dx \\
& = \frac{M}{\left(D\hspace{-1mm}+\hspace{-1mm}D_0\right)}\sum\nolimits_{n = 1}^N \hspace{-1mm}{{\beta _n}{e^{\frac{{{c_n M N_0}}}{{P_{\text{max}}}}}} {E_1}\left( {\frac{{{c_n M N_0}}}{{P_{\text{max}}}}} \right)}.\notag
\end{align}

\subsection{ESG in Multi-antenna Systems}
Based on \eqref{ErgodicSumRateMIMONOMA} and \eqref{ErgodicSumRateMIMOOMAZF}, the asymptotic ESG of MIMO-NOMA over MIMO-OMA with ${K \rightarrow \infty}$ can be obtained by:
\begin{align}\label{EPGMIMOERA}
\hspace{-3mm}\mathop {\lim }\limits_{K \rightarrow \infty}  \overline{G^{\rm{MIMO}}} &\hspace{-1mm}=
\mathop {\lim }\limits_{K \to \infty }  \overline{R_{\rm{sum}}^{{\rm{MIMO-NOMA}}}} - \overline{R_{\rm{sum}}^{{\rm{MIMO-OMA}}}} \notag\\
&\hspace{-1mm}\approx M{\ln}\left( {1 + \frac{{P_{\text{max}}}}{{{\left(D+D_0\right)N_0}}}\sum\nolimits_{n = 1}^N {\frac{{{\beta _n}}}{{{c_n}}}} } \right) \notag\\
\hspace{-3mm}&\hspace{-1mm}- \frac{M}{\left(D\hspace{-1mm}+\hspace{-1mm}D_0\right)}\hspace{-1mm}\sum\nolimits_{n = 1}^N \hspace{-1mm}{{\beta _n}{e^{\frac{{{c_n}MN_0}}{{P_{\text{max}}}}}} \hspace{-1mm}{E_1}\hspace{-1mm}\left( \hspace{-0.5mm} {\frac{{{c_n}MN_0}}{{P_{\text{max}}}}} \hspace{-0.5mm} \right)}.
\end{align}
To reveal more insights, similar to equation \eqref{EPGSISOERA2}, we consider the asymptotic ESG of MIMO-NOMA over MIMO-OMA in the high SNR regime as follows
\vspace{-1mm}
\begin{equation}\label{EPGMIMOERA2}
\hspace{-2mm}\mathop {\lim }\limits_{K \rightarrow \infty, {P_{\text{max}}} \rightarrow \infty}  \overline{G^{\rm{MIMO}}} \approx M\vartheta \left( {D,{D_0}} \right) + M\ln\left(M\right) + M\gamma,\hspace{-0.5mm}
\vspace{-1mm}
\end{equation}
where $\vartheta \left( {D,{D_0}} \right)$ is given in \eqref{NearFarDiversity}.

Comparing \eqref{EPGSISOERA2} and \eqref{EPGMIMOERA2}, we have
\vspace{-1mm}
\begin{equation}\label{EPGMIMOERA3}
\mathop {\lim }\limits_{K \rightarrow \infty, {P_{\text{max}}} \rightarrow \infty}  \overline{G^{\rm{MIMO}}} = M\mathop {\lim }\limits_{K \rightarrow \infty, {P_{\text{max}}} \rightarrow \infty} \overline{G^{{\rm{SISO}}}} + M\ln\left(M\right),
\vspace{-1mm}
\end{equation}
which implies that the asymptotic ESG of SISO-NOMA over SISO-OMA is amplified by $M$ times when $M$ antennas are deployed at the BS.
In fact, for $K \rightarrow \infty$, the received signals fully span in the $M$-dimensional signal space\cite{WangMUG}, which enables MIMO-NOMA and MIMO-OMA to fully exploit the system spatial DOF, $M$.
In addition, to suppress the inter-user interference in MIMO-OMA, there is a factor of $\frac{1}{M}$ power loss on average within each group due to the ZF projection\cite{Tse2005}.
Therefore, we have an additional power gain $\ln\left(M\right)$ of MIMO-NOMA over MIMO-OMA in the third term in \eqref{EPGMIMOERA3}.

\section{Simulations}

\begin{figure}[t]
\centering
\vspace{-4mm}
\subfigure[$M = 1$]
{\label{APGVsK:a} 
\includegraphics[width=0.23\textwidth]{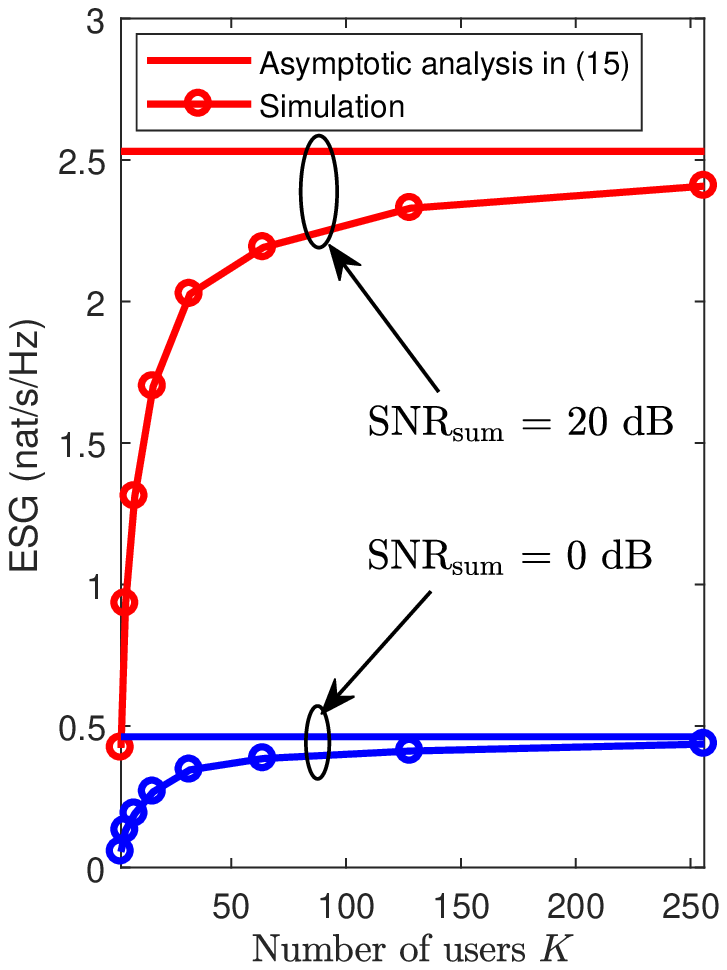}}
\subfigure[$M = 4$]
{\label{APGVsK:b} 
\includegraphics[width=0.23\textwidth]{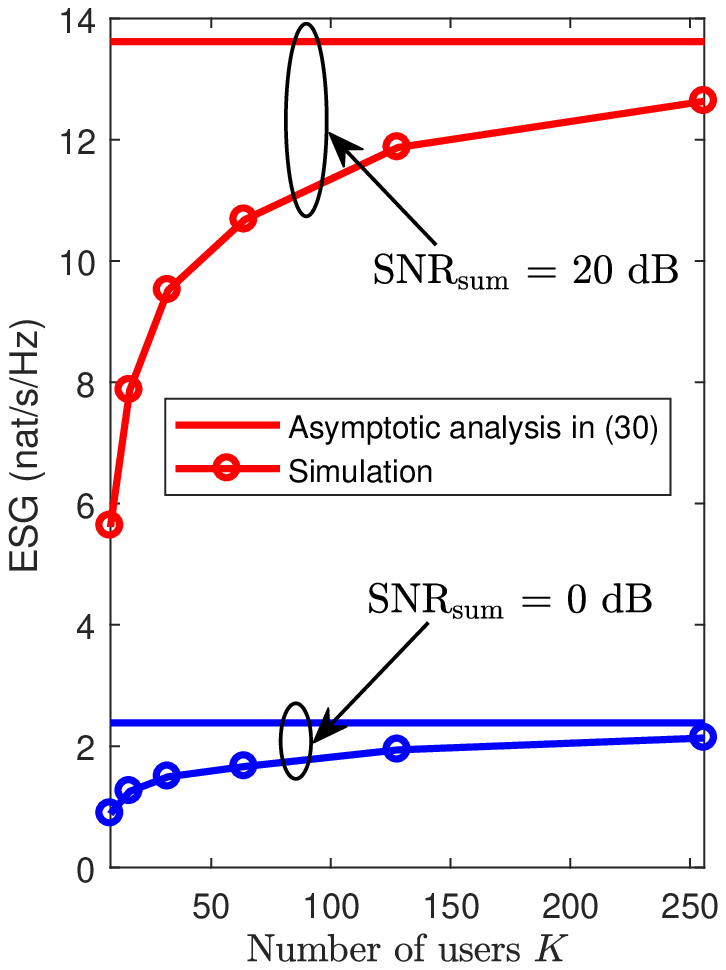}}\vspace{-2mm}
\caption{The ESG of NOMA over OMA with ERA versus the number of users $K$. The inner radius is $D_0 = 50$ m and the outer cell radius is $D = 500$ m. Two simulation cases with ${\rm{SNR}_{sum}} = 20$ dB and ${\rm{SNR}_{sum}} = 0$ dB are considered.}\vspace{-5mm}
\label{APGVsK}%
\end{figure}

We use simulations to verify our derived analytical results.
The inner cell radius is $D_0 = 50$ m and the outer cell radius is given by $D = [50, 200, 500]$ m.
The number of users ranges from $2$ to $256$ and the number of antennas equipped at the BS is $M = [1, 2, 4, 8]$\footnote{In this work, we focus on the multi-antenna system with a finite number of antennas at the BS. The performance gain of NOMA over OMA with a massive-antenna array at the BS will be considered in our future work.}.
The path loss exponent is $\alpha = 3.76$ according to the 3GPP path loss model\cite{Access2010}.
To characterize the system SNR in simulations, we define the received sum SNR at the BS as follows:
\vspace{-1mm}
\begin{equation}\label{SystemSNR}
{\rm{SNR}_{sum}} = \frac{{P_{\text{max}}}}{N_0} {\overline{{{\left| {{{h}}} \right|}^2}}} = \frac{{P_{\text{max}}}}{N_0} \frac{\overline{{{\left\| {{\mathbf{h}}} \right\|}^2}}}{M},
\vspace{-1mm}
\end{equation}
where ${\overline{{{\left| {{{h}}} \right|}^2}}}$ and ${\overline{{{\left\| {{\mathbf{h}}} \right\|}^2}}}$ are given by \eqref{Mean_ChannelPowerGainSISO} and \eqref{Mean_ChannelPowerGainMIMO}, respectively.
The total transmit power ${P_{\text{max}}}$ is adjusted adaptively for different cell sizes to satisfy ${\rm{SNR}_{sum}}$ in \eqref{SystemSNR} ranging from $0$ dB to $40$ dB.
All the simulation results in this paper are averaged over both small-scale fading and large-scale fading.

Fig. \ref{APGVsK} illustrates the ESG of NOMA and OMA versus the number of users.
We consider $M=1$ for SISO-NOMA and SISO-OMA in Fig. \ref{APGVsK:a} and $M=4$ for MIMO-NOMA and MIMO-OMA in Fig. \ref{APGVsK:b}.
In both single-antenna and multi-antenna scenarios, we observe a higher ESG of NOMA over OMA in the high SNR case, e.g. ${\rm{SNR}_{sum}} = 20$ dB.
In addition, for both cases with $M=1$ and $M=4$, the ESGs increase with the number of users $K$ monotonically and approach the derived asymptotic ESG expressions in \eqref{EPGSISOERA} and \eqref{EPGMIMOERA}.
It implies that the ESG saturates with increasing $K$ and thus the near-far gain can be captured by a finite number of users.
Comparing Fig. \ref{APGVsK:a} and Fig. \ref{APGVsK:b}, we can observe that a substantial improvement in ESG when applying NOMA in multi-antenna systems due to the extra spatial DOF as predicted in \eqref{EPGMIMOERA3}.

\begin{figure}[t]
\centering
\vspace{-4mm}
\subfigure[$M = 1$]
{\label{APGVsSNR:a} 
\includegraphics[width=0.23\textwidth]{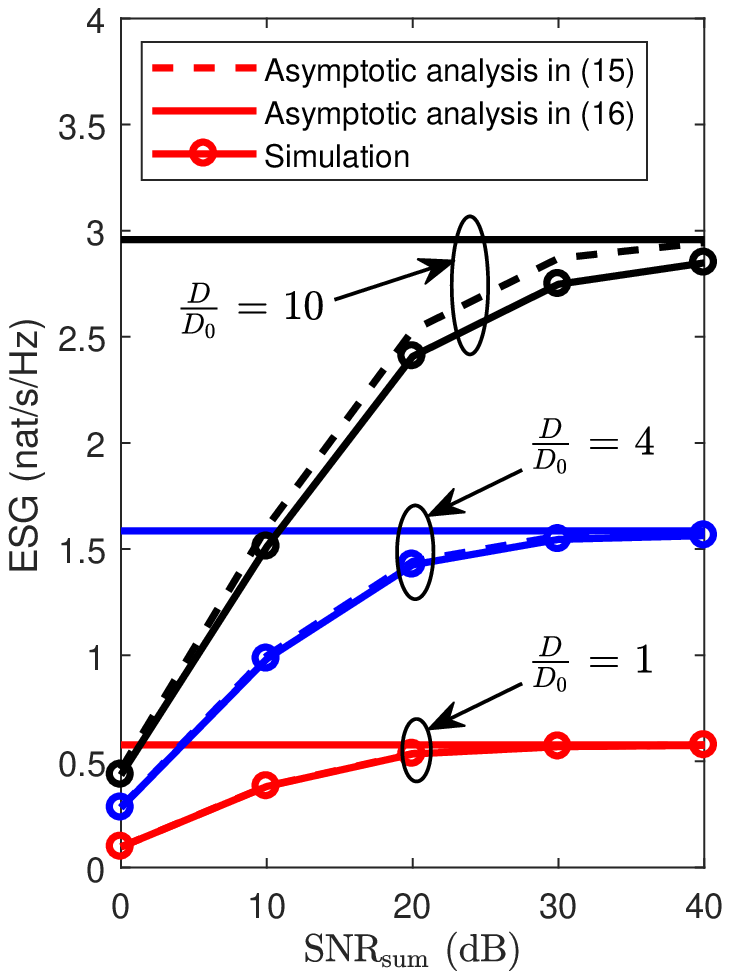}}\vspace{-2mm}
\subfigure[$M = 4$]
{\label{APGVsSNR:b} 
\includegraphics[width=0.23\textwidth]{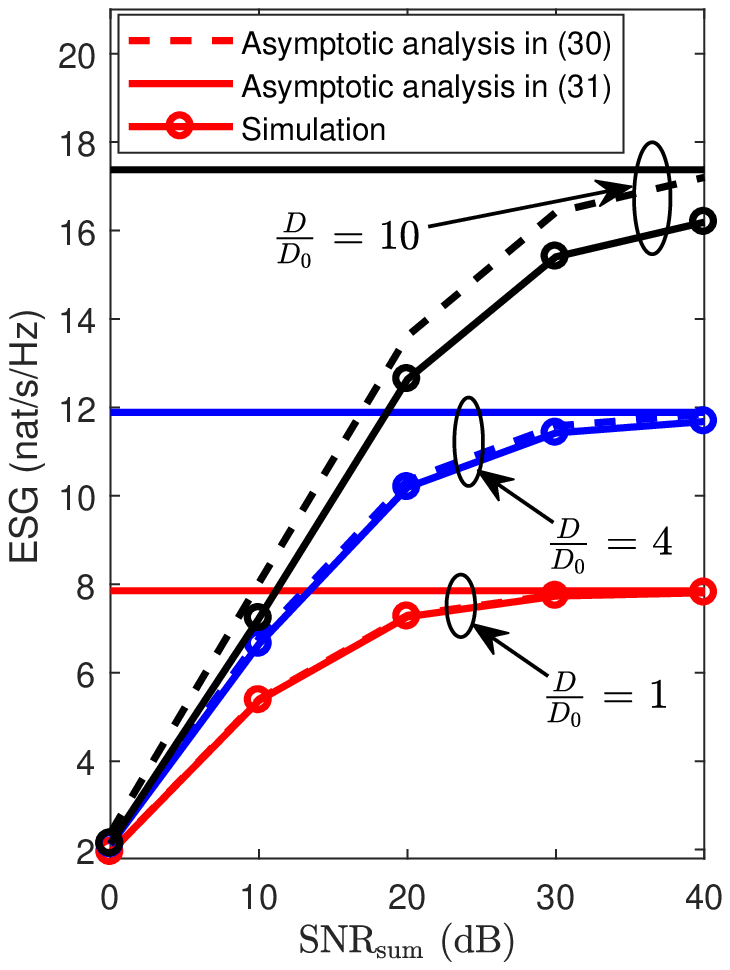}}
\caption{The ESG of NOMA over OMA with ERA versus ${\rm{SNR}_{sum}}$. The number of users is $K = 256$ and three simulation cases with $\frac{D}{D_0} = [1,4,10]$ are considered for comparison.}\vspace{-5mm}
\label{APGVsSNR}%
\end{figure}

Fig. \ref{APGVsSNR} depicts the ESG of NOMA over OMA versus ${\rm{SNR}_{sum}}$, in both single-antenna and multi-antenna scenarios.
We can observe that the asymptotic analyses of ESG in \eqref{EPGSISOERA} and \eqref{EPGMIMOERA} matches simulation results in all the considered cases, particularly for the case with a low $\frac{D}{D_0}$.
With increasing the SNR, the ESGs monotonically approach the asymptotic analyses in \eqref{EPGSISOERA2} and \eqref{EPGMIMOERA2} for the cases of $M=1$ and $M=4$, respectively.
In particular, in single-antenna scenario, when all the users are randomly distributed on a circle with $D = D_0 = 50$ m, we can observe an ESG about $0.575$ nat/s/Hz at ${\rm{SNR}_{sum}} = 40$ dB.
This verifies the accuracy of the derived small-scale near-far gain in \eqref{EPGSISOERA2}.
Besides, in both single-antenna and multi-antenna systems, we can observe that a larger $\frac{D}{D_0}$ results in a larger performance gain owing to the increased large-scale near-far gain.
In addition, it can be observed that the ESG increases faster for the case with a larger $\frac{D}{D_0}$.
In other words, a higher large-scale near-far gain enables NOMA to utilize the power more efficiently.

\begin{figure}[!t]
\centering
\includegraphics[width=3in]{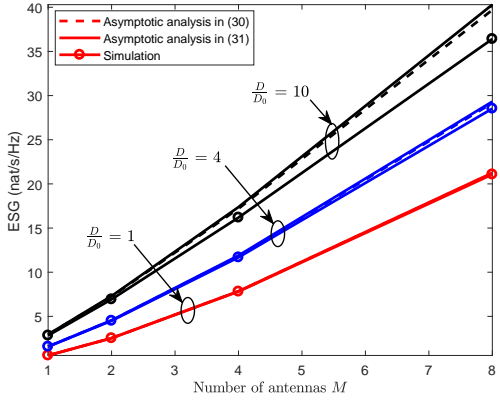}\vspace{-3mm}
\caption{The ESG of NOMA over OMA with ERA versus the number of antennas $M$. The simulation setup is the same as Fig. \ref{APGVsSNR} except ${\rm{SNR}_{sum}} = 40$ dB and $M = [1, 2, 4, 8]$.}\vspace{-4mm}
\label{MIMOAPGVsM}
\end{figure}

Fig. \ref{MIMOAPGVsM} illustrates the ESG of MIMO-NOMA over MIMO-OMA versus the number of antennas $M$.
It can be observed that the simulation results follow our asymptotic analyses derived in \eqref{EPGMIMOERA} and \eqref{EPGMIMOERA2} closely, especially for the case with a small $\frac{D}{D_0}$.
More importantly, as predicted in \eqref{EPGMIMOERA2}, a larger $\frac{D}{D_0}$ enables a larger increasing rate in the ESG with respective to the number of antennas $M$, due to the increased large-scale near-far gain $\vartheta \left( {D,{D_0}} \right)$.

\section{Conclusion and Discussion for Multi-cell Systems}
In this paper, we investigated the ESG brought by NOMA over OMA in both single-antenna and multi-antenna systems via asymptotic performance analyses for a sufficiently large number of users in the high SNR regime.
For single-antenna systems, the ESG of NOMA over OMA was quantified and two types of near-far gains were identified in the derived ESG, i.e., the large-scale near-far gain and the small-scale near-far gain.
The large-scale near-far gain increases with the cell size, while the small-scale near-far gain is a constant of $\gamma = 0.57721$ nat/s/Hz in Rayleigh fading channels.
Furthermore, we unveiled that the ESG of SISO-NOMA over SISO-OMA can be amplified by $M$ times when equipping $M$ antennas at the BS, owing to the extra spatial DOF offered by additional antennas.

In a multi-cell system, the system SNR in \eqref{SystemSNR} should be redefined as follows:
\vspace{-1mm}
\begin{equation}
{\rm{SNR}_{sum}^{multicell}} = \frac{{P_{\text{max}}}}{ \beta{P_{\text{max}}} + N_0} {\overline{{{\left| {{{h}}} \right|}^2}}},
\vspace{-1mm}
\end{equation}
where $\beta$ characterizes the inter-cell interference\cite{Xu2017}.
Due to this interference, NOMA might work at a low SNR regime and the performance gain of NOMA over OMA in multi-cell systems will be considered in our future work.

\section{Acknowledgement}
The authors would like to appreciate Prof. Ping Li from City University of Hong Kong for valuable discussion during this work.


\begin{thebibliography}{10}
\providecommand{\url}[1]{#1}
\csname url@samestyle\endcsname
\providecommand{\newblock}{\relax}
\providecommand{\bibinfo}[2]{#2}
\providecommand{\BIBentrySTDinterwordspacing}{\spaceskip=0pt\relax}
\providecommand{\BIBentryALTinterwordstretchfactor}{4}
\providecommand{\BIBentryALTinterwordspacing}{\spaceskip=\fontdimen2\font plus
\BIBentryALTinterwordstretchfactor\fontdimen3\font minus
  \fontdimen4\font\relax}
\providecommand{\BIBforeignlanguage}[2]{{%
\expandafter\ifx\csname l@#1\endcsname\relax
\typeout{** WARNING: IEEEtran.bst: No hyphenation pattern has been}%
\typeout{** loaded for the language `#1'. Using the pattern for}%
\typeout{** the default language instead.}%
\else
\language=\csname l@#1\endcsname
\fi
#2}}
\providecommand{\BIBdecl}{\relax}
\BIBdecl

\bibitem{Dai2015}
L.~Dai, B.~Wang, Y.~Yuan, S.~Han, I.~Chih-Lin, and Z.~Wang, ``Non-orthogonal
  multiple access for {5G}: solutions, challenges, opportunities, and future
  research trends,'' \emph{IEEE Commun. Mag.}, vol.~53, no.~9, pp. 74--81, Sep.
  2015.

\bibitem{Ding2015b}
Z.~Ding, Y.~Liu, J.~Choi, Q.~Sun, M.~Elkashlan, C.~L. I, and H.~V. Poor,
  ``Application of non-orthogonal multiple access in {LTE} and {5G} networks,''
  \emph{IEEE Commun. Mag.}, vol.~55, no.~2, pp. 185--191, Feb. 2017.

\bibitem{WeiSurvey2016}
Z.~Wei, Y.~Jinhong, D.~W.~K. Ng, M.~Elkashlan, and Z.~Ding, ``A survey of
  downlink non-orthogonal multiple access for {5G} wireless communication
  networks,'' \emph{ZTE Commun.}, vol.~14, no.~4, pp. 17--25, Oct. 2016.

\bibitem{wong2017key}
V.~W. Wong, R.~Schober, D.~W.~K. Ng, and L.-C. Wang, \emph{Key Technologies for
  {5G} Wireless Systems}.\hskip 1em plus 0.5em minus 0.4em\relax Cambridge
  University Press, 2017.

\bibitem{ChenNOMAScheme}
Y.~Chen, A.~Bayesteh, Y.~Wu, B.~Ren, S.~Kang, S.~Sun, Q.~Xiong, C.~Qian, B.~Yu,
  Z.~Ding, S.~Wang, S.~Han, X.~Hou, H.~Lin, R.~Visoz, and R.~Razavi, ``Toward
  the standardization of non-orthogonal multiple access for next generation
  wireless networks,'' \emph{IEEE Commun. Mag.}, vol.~56, no.~3, pp. 19--27,
  Mar. 2018.

\bibitem{DerrickEEOFDMA}
D.~W.~K. Ng, E.~S. Lo, and R.~Schober, ``Energy-efficient resource allocation
  in {OFDMA} systems with large numbers of base station antennas,'' \emph{IEEE
  Trans. Wireless Commun.}, vol.~11, no.~9, pp. 3292--3304, Sep. 2012.

\bibitem{Sun2016Fullduplex}
Y.~Sun, D.~W.~K. Ng, Z.~Ding, and R.~Schober, ``Optimal joint power and
  subcarrier allocation for full-duplex multicarrier non-orthogonal multiple
  access systems,'' \emph{IEEE Trans. Commun.}, vol.~65, no.~3, pp. 1077--1091,
  Mar. 2017.

\bibitem{WeiTCOM2017}
Z.~Wei, D.~W.~K. Ng, J.~Yuan, and H.~M. Wang, ``Optimal resource allocation for
  power-efficient {MC-NOMA} with imperfect channel state information,''
  \emph{IEEE Trans. Commun.}, vol.~65, no.~9, pp. 3944--3961, May 2017.

\bibitem{Qiu2018Lattice}
M.~Qiu, Y.~C. Huang, S.~L. Shieh, and J.~Yuan, ``A lattice-partition framework
  of downlink non-orthogonal multiple access without {SIC},'' \emph{IEEE Trans.
  Commun.}, pp. 1--1, Jan. 2018.

\bibitem{Ding2014}
Z.~Ding, Z.~Yang, P.~Fan, and H.~Poor, ``On the performance of non-orthogonal
  multiple access in {5G} systems with randomly deployed users,'' \emph{IEEE
  Signal Process. Lett.}, vol.~21, no.~12, pp. 1501--1505, Dec. 2014.

\bibitem{Chen2017}
Z.~Chen, Z.~Ding, X.~Dai, and R.~Zhang, ``An optimization perspective of the
  superiority of {NOMA} compared to conventional {OMA},'' \emph{IEEE Trans.
  Signal Process.}, vol.~65, no.~19, pp. 5191--5202, Oct. 2017.

\bibitem{WangPowerEfficiency}
P.~Wang, J.~Xiao, and L.~P, ``Comparison of orthogonal and non-orthogonal
  approaches to future wireless cellular systems,'' \emph{IEEE Veh. Technol.
  Mag.}, vol.~1, no.~3, pp. 4--11, Sep. 2006.

\bibitem{Xu2017}
C.~Xu, Y.~Hu, C.~Liang, J.~Ma, and L.~Ping, ``Massive {MIMO}, non-orthogonal
  multiple access and interleave division multiple access,'' \emph{IEEE
  Access}, vol.~5, pp. 14\,728--14\,748, Jul. 2017.

\bibitem{DingSignalAlignment}
Z.~Ding, R.~Schober, and H.~V. Poor, ``A general {MIMO} framework for {NOMA}
  downlink and uplink transmission based on signal alignment,'' \emph{IEEE
  Trans. Wireless Commun.}, vol.~15, no.~6, pp. 4438--4454, Jun. 2016.

\bibitem{ChenQuasiDegradation}
Z.~Chen, Z.~Ding, X.~Dai, and G.~K. Karagiannidis, ``On the application of
  quasi-degradation to {MISO-NOMA} downlink,'' \emph{IEEE Trans. Signal
  Process.}, vol.~64, no.~23, pp. 6174--6189, Dec. 2016.

\bibitem{Vishwanath2003}
S.~Vishwanath, N.~Jindal, and A.~Goldsmith, ``Duality, achievable rates, and
  sum-rate capacity of gaussian {MIMO} broadcast channels,'' \emph{IEEE Trans.
  Inf. Theory}, vol.~49, no.~10, pp. 2658--2668, Oct. 2003.

\bibitem{WangMUG}
P.~Wang and L.~Ping, ``On maximum eigenmode beamforming and multi-user gain,''
  \emph{IEEE Trans. Inf. Theory}, vol.~57, no.~7, pp. 4170--4186, Jul. 2011.

\bibitem{wei2017performance}
Z.~Wei, L.~Dai, D.~W.~K. Ng, and J.~Yuan, ``Performance analysis of a hybrid
  downlink-uplink cooperative {NOMA} scheme,'' in \emph{Proc. IEEE Veh. Techn.
  Conf.}, Jun. 2017, pp. 1--7.

\bibitem{Tse2005}
D.~Tse and P.~Viswanath, \emph{Fundamentals of wireless communication}.\hskip
  1em plus 0.5em minus 0.4em\relax Cambridge university press, 2005.

\bibitem{WeiLetter2018}
Z.~Wei, D.~W.~K. Ng, and J.~Yuan, ``Joint pilot and payload power control for
  uplink {MIMO-NOMA} with {MRC-SIC} receivers,'' \emph{IEEE Commun. Lett.},
  vol.~PP, no.~99, pp. 1--1, Jan. 2018.

\bibitem{abramowitz1964handbook}
M.~Abramowitz and I.~A. Stegun, \emph{Handbook of mathematical functions: with
  formulas, graphs, and mathematical tables}.\hskip 1em plus 0.5em minus
  0.4em\relax Courier Corporation, 1964, vol.~55.

\bibitem{kedlaya1994proof}
K.~Kedlaya, ``Proof of a mixed arithmetic-mean, geometric-mean inequality,''
  \emph{The American Mathematical Monthly}, vol. 101, no.~4, pp. 355--357,
  1994.

\bibitem{Dingtobepublished}
Z.~Ding, P.~Fan, and H.~V. Poor, ``Impact of user pairing on {5G} nonorthogonal
  multiple-access downlink transmissions,'' \emph{IEEE Trans. Veh. Technol.},
  vol.~65, no.~8, pp. 6010--6023, Aug. 2016.

\bibitem{GoldsmithMIMOCapacity2003}
A.~Goldsmith, S.~A. Jafar, N.~Jindal, and S.~Vishwanath, ``Capacity limits of
  {MIMO} channels,'' \emph{IEEE J. Select. Areas Commun.}, vol.~21, no.~5, pp.
  684--702, Jun. 2003.

\bibitem{yang2017noma}
Q.~Yang, H.-M. Wang, D.~W.~K. Ng, and M.~H. Lee, ``{NOMA} in downlink {SDMA}
  with limited feedback: Performance analysis and optimization,'' \emph{IEEE J.
  Select. Areas Commun.}, vol.~35, no.~10, pp. 2281--2294, 2017.

\bibitem{Access2010}
``Evolved universal terrestrial radio access: Further advancements for {E-UTRA}
  physical layer aspects,'' 3GPP TR 36.814, Tech. Rep., 2010.

\end{thebibliography}


\end{document}